\newif\ifdraft \drafttrue
\newif\iffull \fulltrue

\documentclass[11pt]{article}
\usepackage{mmpharvard}
\usepackage{fullpage}
\usepackage{setspace}

\newcommand{\ignore}[1]{}

\newcommand{\ba}{\mathbf{a}}
\newcommand{\bh}{\mathbf{h}}
\renewcommand{\Re}{\mathbb{R}}
\title{\textsc{An Anti-Folk Theorem for Large Repeated Games with Imperfect Monitoring}\thanks{We would like to especially thank Yu Awaya and Vijay Krishna for several helpful comments regarding the manuscript.}}
\author{%
\textsc{Mallesh M. Pai}\thanks{Department of Economics, University of Pennsylvania. Email: \href{mailto:mallesh@econ.upenn.edu}{mallesh@econ.upenn.edu}.}
\and
\textsc{Aaron Roth}\thanks{Department of Computer and Information Science, University of Pennsylvania. Email: \href{mailto:aaroth@cis.upenn.edu}{aaroth@cis.upenn.edu} Supported in part by NSF Grant CCF-1101389, and NSF CAREER grant, and a Google Research Grant.}
\and
\textsc{Jonathan Ullman}\thanks{Center for Research on Computation and Society, Harvard University. Email: \href{mailto:jullman@seas.harvard.edu}{jullman@seas.harvard.edu}}
}

\date{\textsc \today}

\begin{document}
\maketitle

\begin{abstract}
We study infinitely repeated games in settings of imperfect monitoring. We first prove a family of theorems that show that when the signals observed by the players satisfy a condition known as $(\epsilon, \gamma)$-differential privacy, that the folk theorem has little bite: for values of $\epsilon$ and $\gamma$ sufficiently small, for a fixed discount factor, any equilibrium of the repeated game involve players playing approximate equilibria of the stage game in every period. Next, we argue that in large games ($n$ player games in which unilateral deviations by single players have only a small impact on the utility of other players), many monitoring settings naturally lead to signals that satisfy $(\epsilon,\gamma)$-differential privacy, for $\epsilon$ and $\gamma$ tending to zero as the number of players $n$ grows large. We conclude that in such settings, the set of equilibria of the repeated game collapse to the set of equilibria of the stage game.

Our results nest and generalize previous results of \citeasnoun{green1980noncooperative}, \citeasnoun{sabourian1990anonymous} and suggest that differential privacy is a natural measure of the ``largeness'' of a game. Further, techniques from the literature on differential privacy allow us to prove quantitative bounds, where the existing literature focuses on limiting results.
\end{abstract}

\thispagestyle{empty}
\newpage
\clearpage
\setcounter{page}{1}

\section{Introduction}
In a repeated game, agents interact with one another repeatedly, and can base their actions on their memory of past interactions.  One of the most robust features of models of infinitely repeated games is the multiplicity of equilibria. Repeatedly playing Nash equilibria of the underlying stage game is preserved as an equilibrium of the repeated game, but a much larger set is introduced as well. The so called ``folk theorems'' state, informally, that any individually rational payoffs for each player (i.e. payoffs that are above their minmax value in the game) can be achieved as an equilibrium of the repeated game (if players are patient enough). The folk theorems can be interpreted as either positive or negative statements.  From the perspective of an optimist (or a mechanism designer), they mean that high social-welfare behaviors (like cooperation in the Prisoner's Dilemma) can be supported as equilibria in repeated games, even though they cannot be supported as equilibria of the stage game. From the perspective of a pessimist (or a computer scientist), they mean that the social welfare of the \emph{worst} equilibrium can get worse (i.e. the price of anarchy increases). Regardless of the stance one takes, because of the severe multiplicity problem, Nash equilibria of repeated games lose much of their predictive power.

The simplest model of repeated games are repeated games with \emph{perfect monitoring}.  In these games, after each period, each agent observes the exact action played by each of his opponents.  In several settings of interest, this assumption is unrealistic.  A more natural model is repeated games with \emph{imperfect monitoring}, in which agents receive some kind of noisy signal about the play at the last period. This signal could be an estimate of the payoff of each of a player's actions, a noisy histogram of how many players played each action, or information that is encoded as a set of prices.  The game could be one of \emph{private monitoring}, in which each player receives his own signal privately, or \emph{public monitoring}, in which every player observes the same signal. Folk theorems are known to hold in almost all models of repeated games, even those with imperfect public or private monitoring.  See \citeasnoun{fudenberg1994folk} for a classic reference in the case of public monitoring, \citeasnoun{sugaya2011folk} for a recent result in the case of private monitoring, and \citeasnoun{mailath2006repeated} for a textbook treatment.

In this paper, we show that these theorems will often lack bite in ``large'' games of imperfect monitoring.  Particularly, we consider games where the observed signals satisfy the notion of ``differential privacy''~\cite{DMNS06}, which we show arises naturally in many games of imperfect monitoring with a large number of players.  We show that in games with differentially private signals, the set of equilibria of the repeated games must involve play of approximate equilibria of the stage game in every period---approximate Nash equilibrium in the case of public monitoring and approximate correlated equilibrium in the case of private monitoring.  Our results therefore suggests that in games with a large number of players, we should not expect to see ``folk theorem equilibria.'' This also adds to a long standing question in the literature on repeated games---for a given fixed $\delta$ what payoffs can be achieved by equilibria of the repeated game?%
\footnote{See e.g. \citeasnoun{mailath2002maximum} for a characterization in the case of the Prisoner's Dilemma with perfect monitoring.}
We implicitly provide an upper bound on this set---it is the set of payoffs that can be achieved by the play of approximate stage game equilibria in every period, with the degree of approximation being bounded by our theorems.

A different way to view our results arises from a seeming discontinuity between dynamic models with a continuum of players and those with a large but finite number of players.%
\footnote{We thank Nabil Al-Najjar and Preston McAfee for pointing out this connection to us.}
In a dynamic model, incentives are generated by agents punishing or rewarding others for past play. When only aggregate play is observed, then the play of one single agent does not affect the observed outcome in the continuum case, and hence individual deviations cannot be detected. With a finite number of players, a deviation may be small, but it is perfectly observable. Therefore deviations can be punished or rewarded. A small literature attempts to bridge this discontinuity by postulating a model of noisy observation. Our paper can be thought of as a ``quantitative'' version of their limit results for repeated games. In particular, the noisiness of our signals are parametrized by two numbers, $\epsilon$ and $\gamma$. We show that as  $\epsilon + \gamma$ vanishes, we approximate the unobservability of the continuum setting, while the perfect observability setting corresponds to $\epsilon + \gamma$ large.

\subsection{Related Work}
It has long been understood (hence the term ``folk theorem'') that the intertemporal incentives generated by repeated interaction can support more outcomes than just the equilibria of the one-off interaction.  The literature is too large to survey here, we refer the interested reader to the textbook by \citeasnoun{mailath2006repeated} for a comprehensive and up-to-date treatment. We review concepts that are standard in the literature as needed in Section \ref{sec:prelim}. Repeated games have also been studied in the computer science literature.
An older literature studies repeated games when agents have limited computational resources---see e.g. \citeasnoun{megiddo1986play}, \citeasnoun{papadimitriou1994complexity} or \citeasnoun{abreu1988structure}. There has also been some recent interest in repeated games in the computer science community. These consider games of perfect monitoring, and study the complexity of computing an equilibrium of the repeated game---see e.g. \citeasnoun{borgs2010myth} and \citeasnoun{halperntruth}.

The idea that noisy observation may limit the equilibria in repeated games has some precedent in the literature---notably the papers of  \citeasnoun{green1980noncooperative}, \citeasnoun{sabourian1990anonymous} and \citeasnoun{al2001large}.%
\footnote{\citeasnoun{levine1995agents} study similar questions in dynamic games.}
In relation to our setting, note that the first paper restricts attention to trigger strategies, the first two papers assume the stage game is anonymous, and all three consider public equilibria in a public monitoring setting. By contrast our results apply to general games that satisfy the differential privacy condition, and apply also to private monitoring settings. The results in those papers are limit theorems, while we prove quantitative bounds on the relationship between the ``amount'' of noise and the degree of deviation from stage-Nash equilibria that can be supported. Finally, while their conditions on the monitoring structure  
may be hard to verify, we are able we use techniques from the literature on differential privacy (see below) to give examples of ``natural'' settings satisfying our conditions.

``Differential Privacy'' is a condition that a family of probability distributions (which are parameterized by $n$ agent ``reports'') may satisfy, which provides a multiplicative bound on the influence that any single report can have on the resulting distribution. It was first defined by \citeasnoun{DMNS06} in the computer science literature, and has become a standard ``privacy'' solution concept. We do not use the condition here to deal with privacy concerns, but instead just as a useful Lipschitz condition that is commonly satisfied by many noisy families of signal distributions. The connection between differential privacy and game theory was first noted by \citeasnoun{MT07} who observed that when this condition is applied to an auction mechanism, it can be used to argue approximate truthfulness. They use it to derive prior-free near revenue-optimal mechanisms for digital goods auctions. Since then, differential privacy has been used as a tool in a number of game theoretic settings, including mechanism design without money \cite{NST12} and mechanism design in prior-free settings of incomplete information with mechanisms that have extremely weak powers of enforcement \cite{kearns2014mechanism,RR13}. A related line of work seeks to use traditional tools in mechanism design in settings in which agents have preferences for ``privacy'', which are quantified by measures related to differential privacy \cite{GR11,NOS12,Xiao13,CCKMV13,GL13,NVX14}. See \citeasnoun{PR13} for a survey of this area and its connections to mechanism design and game theory.

\section{Preliminaries}\label{sec:prelim}
We consider infinitely repeated games that are played over a series of periods. Time is discrete and is indexed by $t = 0,1,2, \ldots$.

\paragraph{Stage Game}
There are $n$  players who repeatedly play a stage game $G$. Each player $i$ has a finite set of actions $A_i$, and gets payoff $u_i: \prod_{j=1}^n A_j \rightarrow [0,1]$. We will denote an action profile (i.e. a vector of actions, one for each player) by $\ba \in \prod_{j=1}^n A_j.$ Mixed strategies are probability distributions over pure strategies, which we write as $\Delta A_i$ for player $i$. We write $\alpha_i \in \Delta A_i$ to denote a mixed strategy for player $i$, and $\mathbb{\alpha}$ to denote a vector of mixed strategies. Payoffs for mixed strategies are, as usual, the expected payoff when the profile of actions is drawn according to the distribution specified by the mixed strategies.

All players discount the future at a rate $\delta \in [0,1)$.  Therefore an annuity that pays $1$ in every period, forever, has a present discounted value of $\frac{1}{1-\delta}$. As is standard in the literature on repeated games, we normalize all payoffs by $(1-\delta)$---the value of the annuity of $1$ forever is therefore normalized to $1$, while e.g. a payoff of $1$ today and nothing thereafter is normalized to $(1-\delta)$ (because the agent would be indifferent between receiving $1$ this period and nothing thereafter and an annuity of $(1-\delta)$).

Before we describe repeated games, we begin by reviewing two approximate equilibrium concepts for the stage game.

\begin{definition} \label{def:ne}
An $\eta$-approximate Nash equilibrium of the stage game $G$ is a vector $\alpha \in \prod_{j=1}^n \Delta A_i$ of mixed strategies, such that for all $i$ and all $a_i \in A_i$:
$$u_i(\alpha) \geq u_i(a_i, \alpha_{-i}) - \eta$$
\end{definition}

\begin{definition} \label{def:ce}
An $\eta$-approximate correlated equilibrium of the stage game $G$ is defined by a set of signals $S$, a distribution $\mathcal{D} \in \Delta S^n$ over $n$-tuples of signals, and a function $\sigma_i:S\rightarrow \Delta A_i$ mapping signals to mixed strategies for each player $i$. $(\mathcal{D}, \sigma)$ is an $\eta$-approximate correlated equilibrium if for every signal $s \in S$ and every action $a_i \in A_i$:
$$\sum_{s_{-i}} \left[u_i(a_i, \sigma_{-i}(s_{-i}))- u_i(\sigma_i(s), \sigma_{-i}(s_{-i}))\right] \mathcal{D}(s_{-i}|s_i = s) \leq \eta$$
\end{definition}


\paragraph{Monitoring}
We consider models of both public and private monitoring. In a game of \emph{public monitoring}, at the end of every period, after players have played their actions for that period, they receive feedback about the actions taken by others via a \emph{public signal}. Formally, at the end of the period, when the profile of actions taken is $\ba$, all players commonly learn a signal $s \in S$, where $s$ is drawn independently by a distribution $P_{\ba} \in \Delta S.$%
\footnote{Note that this model subsumes the standard perfect monitoring model---if $S \equiv \prod A_i$ and $P_{\ba}(s) = 1 \text{ iff } s = \ba$.}

In a game of \emph{private} monitoring, at the end of every period, after players have played their actions for that period, each player $i$ receives a private signal $s_i \in S$. Player $i$ does not see the signals $s_j$ for $j \neq i$. Formally, at the end of a period in which the profile of actions taken is $\ba$, a vector of signals $s \in S^n$ is drawn from a distribution $P_{\ba} \in \Delta S^n$,
and each player $i$ observes component $s_i$ of $s$. Note that the signals that different players receive can be correlated. Further, note that this generalizes the public monitoring setting---public monitoring is simply the case where every player's signal is identical.

The ex-post payoff of each player $i$ is given by a function $U_i: A_i \times S \rightarrow \Re$ so that the signal a player receives is a sufficient statistic for the payoff he receives. We assume that
\begin{align}\label{eqn:noisy_payoffs}
u_i(a_i, a_{-i}) = \sum_{s \in S} U_i(a_i, s) P_{\ba}(s),
\end{align}
so that the noisy payoff does not affect the incentive structure of the game.

We consider a commonly satisfied condition on the distribution of signals that players observe---informally, that they are \emph{noisy} in a way that does not reveal ``much'' about the action of any single player.

\begin{definition}
A public monitoring signal structure is said to satisfy $(\epsilon,\gamma)$-differential privacy for some $\epsilon, \gamma >0$ if for every player $i$ and for every $\ba =  (a_1,a_2, \ldots, a_n)$, $\ba' = (a_i', a_{-i})$, and any event $E \subseteq S,$
\begin{align*}
\exp(-\epsilon) P_{\ba'}(E)-\gamma\leq P_{\ba}(E) \leq \exp(\epsilon) P_{\ba'}(E)+\gamma.
\end{align*}
A private monitoring signal structure satisfies $(\epsilon,\gamma)$-differential privacy if for every player $i$ and for every $\ba =  (a_1,a_2, \ldots, a_n)$, $\ba' = (a_i', a_{-i})$, and for every event $E \subseteq S^{n}$,
\begin{align*}
\exp(-\epsilon) P_{\ba'}(E)-\gamma\leq P_{\ba}(E) \leq \exp(\epsilon) P_{\ba'}(E)+\gamma.
\end{align*}
\end{definition}
\begin{remark}
Differential privacy is a condition applied to study the privacy of randomized algorithms on a ``database''(i.e. not generally applied to the signal structure of a repeated games), and was first defined by Dwork, McSherry, Nissim, and Smith \cite{DMNS06}.
\end{remark}
\begin{remark}
Due to (\ref{eqn:noisy_payoffs}), the assumption that a signalling structure satisfies $(\epsilon,\gamma)$-differential privacy is nontrivial. Fixing values of $\epsilon, \gamma$, there exist stage games for  which no such $(\epsilon, \gamma)$-differentially private signal structure can satisfy (\ref{eqn:noisy_payoffs})). That said, several stage games of interest can be paired with differentially private signal structures. Broadly, this includes any \emph{large} game, i.e. one in which any player, by unilaterally changing the action he plays, has a small impact on the stage-game utility of any other player. We present several examples in Section \ref{sec:examples}.
\end{remark}

\begin{remark}
The papers of \citeasnoun{green1980noncooperative} and \citeasnoun{sabourian1990anonymous} use the total variation norm, which corresponds to $(0,\gamma)$-differential privacy.
\end{remark}

For technical reasons, the following definition is often assumed (see e.g. \citeasnoun{mailath2006repeated}). We say the  that the monitoring structure is one of \emph{no observable deviations}, i.e. that the marginal distribution of private signals have full support regardless of the action profile played. Formally:
\begin{definition} \label{def:nod}
A repeated game of public monitoring is said to have \emph{no observable deviations} if $$\forall\,\ba \in \prod A_i, s \in S: \;\; P_{\ba}(s) >0.$$
A repeated game of private monitoring is said to have no observable deviations if $$\forall\, \ba \in \prod A_i, s_i \in S:\;\; P_{\ba}(s_i) >0.$$
\end{definition}

\bigskip

Intuitively, in the repeated game agents will be trying to infer what other players played from this signal---this assumption ensures that Bayesian updating is well defined after any observed signal.

\paragraph{Histories and Strategies}
In the public monitoring setting the public information in period $t$ is the history of public signals, $h^t = (s^1,s^2, \ldots, s^{t-1})$. The set of public histories is
\begin{align*}
\mathscr{H} \equiv \cup_{t=2}^\infty S^{t-1}.
\end{align*}
A history for a player includes both the public history and the history of actions he has taken, $h^t_i = (s^1, a^1_i, s^2, a^2_i, \ldots s^{t-1}, a^{t-1}_i).$ Given a a private history $h^t_i= (s^1, a^1_i, s^2, a^2_i, \ldots s^{t-1}, a^{t-1}_i)$, one can define the public history $h^t = (s^1,s^2, \ldots s^{t-1})$ as just the vector of public signals observed. The set of private histories for player $i$ is
\begin{align*}
\mathscr{H}_i \equiv \cup_{t=2}^{\infty} (A_i \times S)^{t-1}.
\end{align*}
We refer to the vector of private histories (one for each player), at period $t$ by $\bh^t \equiv (h_1^t, h_2^t, \ldots, h_n^t).$
In the private monitoring setting, the history for a player $i$ is the history of the actions he has taken, together with his sequence of private signals: $h^t_i = (s^1_i, a^1_i, s^2_i, a^2_i, \ldots s^{t-1}_i, a^{t-1}_i).$

A pure strategy for player $i$ is a mapping from the set of histories of player $i$ to the set of actions of player $i$:
\begin{align*}
\sigma_i : \mathscr{H}_i \to A_i,
\end{align*}
and a mixed strategy is a mapping from histories to distributions over actions. We denote $\sigma_i(h_i^t)(a_i)$ as the probability strategy  $\sigma_i$ puts on action $a_i$ after private hisotry $h_i^t$.

Strategies in a repeated game can in general depend on a player's entire history, but in a game of public monitoring, we can also consider \emph{public} strategies which depend only on the public signals observed so far (and not on the unobserved actions played by the player).
\begin{definition}\label{def:public_strat}
A strategy $\sigma_i$ is public if it only depends on the public history $\mathscr{H}$, i.e. for any period $t$  and two private histories $h^t_i, h^{\prime t}_i$ with the same public history $h^t$,
\begin{align*}
\sigma_i(h^t_i) = \sigma_i(h^{\prime t}_i).
\end{align*}
\end{definition}
%

\paragraph{Solution Concepts}
Equilibria are defined with respect to the cumulative discounted payoff that players obtain when the play according to a profile of strategies $\sigma$. Formally, we need to recursively define the expected payoff of strategy profile $\sigma$ to player $i$ after private history $h_i^t$. Let us begin with the case of public monitoring.

Fix a strategy profile $\sigma$. Then, for every player $i$ and every history $h_i^t$, we can define a posterior distribution $\mathcal{D}(i, h_i^t, \sigma_{-i})$ on histories $h_{-i}^t$ observed by the other players. In other words the probability that the other players observed private histories $h_{-i}^t$ given player $i$'s private history is $h_i^t$ and they play strategies $\sigma_{-i}$ is given as $\mathcal{D}(i, h_i^t, \sigma_{-i})(h_{-i}^t)$.

Further, given strategies $\sigma$, define $W_{i,\sigma}(\bh^t)$ as:
\begin{align*}
W_{i,\sigma}(\bh^t) =& (1-\delta) u_i(\sigma_i(h_i^t), \sigma_{-i}(h_{-i}^t)) \\
&+ \delta \sum_{s,a_i, a_{-i}} W_{i,\sigma}((\bh^t, s, (a_i,a_{-i}))) \sigma_i(h_i^t)(a_i) \sigma_{-i}(h_{-i}^t)(a_{-i}) P_{(a_i,a_{-i})}(s)
\end{align*}
Note that the argument to $W_i$ is the vector of private histories of all agents $\bh^t$. It represents the expected discounted value to agent $i$ continuing from a point where his private history is $h_i^t$ \emph{and} the private histories of other players are $h_{-i}^t$. Of course player $i$ only knows his own private history, therefore player $i$'s continuation value after a private history $h_i^t$ is
\begin{align*}
V_{i,\sigma}(h_i^t) =& \sum_{h_{-i}^t} W_{i,\sigma}((h_i^t,h_{-i}^t))  \mathcal{D}(i, h_i^t, \sigma_{-i})(h_{-i}^t)\\
=& (1- \delta) \sum_{h_{-i}^t}u_i(\sigma_i(h_i^t), \sigma_{-i}(h_{-i}^t)) \mathcal{D}(i, h_i^t, \sigma_{-i})(h_{-i}^t)\\
&+ \delta \sum_{h_{-i}^t,s,a_i, a_{-i}} \bigg( W_{i,\sigma}(((h_i^t, h_{-i}^t, s, (a_i,a_{-i})))\mathcal{D}(i,h_i^t,\sigma_{-i})(h_{-i}^t) \\
& \hphantom{+ \delta \sum_{h_{-i}^t,s_i, a_i, s_{-i}, a_{-i}} \bigg(\,\, }\;\; \sigma_i(h_i^t)(a_i) \sigma_{-i}(h_{-i}^t)(a_{-i}) P_{(a_i,a_{-i})}(s) \bigg)
\end{align*}

For the case that $\sigma$ is in public strategies, for any public history $h^t$, we can simplify the continuation value as:
\begin{align*}
V_{i,\sigma}(h^t) = (1-\delta) u_i(\sigma(h^t)) + \delta \sum_{s} V_i(\sigma(h^t,s)) P_{\sigma(h^t)}(s).
\end{align*}

%

For the case of private monitoring, the definition is slightly different, since every player sees a separate private signal. Define $W_{i,\sigma}(\bh^t)$ as:
\begin{align*}
W_{i,\sigma}(\bh^t) =& (1-\delta) u_i(\sigma_i(h_i^t), \sigma_{-i}(h_{-i}^t)) \\
&+ \delta \sum_{s_i,s_{-i},a_i, a_{-i}} W_{i,\sigma}((\bh^t, (s_i,s_{-i}), (a_i,a_{-i}))) \sigma_i(h_i^t)(a_i) \sigma_{-i}(h_{-i}^t)(a_{-i}) P_{(a_i,a_{-i})}[s_i,s_{-i}]
\end{align*}
Player $i$'s continuation value after a private history $h_i^t$ is
\begin{align*}
V_{i,\sigma}(h_i^t) =& \sum_{h_{-i}^t} W_{i,\sigma}((h_i^t,h_{-i}^t))  \mathcal{D}(i, h_i^t, \sigma_{-i})(h_{-i}^t)\\
=& (1- \delta) \sum_{h_{-i}^t}u_i(\sigma_i(h_i^t), \sigma_{-i}(h_{-i}^t)) \mathcal{D}(i, h_i^t, \sigma_{-i})(h_{-i}^t)\\
&+ \delta \sum_{h_{-i}^t,s_i,s_{-i},a_i, a_{-i}} \bigg( W_{i,\sigma}(((h_i^t, h_{-i}^t, (s_i,s_{-i}), (a_i,a_{-i})))\mathcal{D}(i,h_i^t,\sigma_{-i})(h_{-i}^t) \\
& \hphantom{+ \delta \sum_{h_{-i}^t,s_i, a_i, s_{-i}, a_{-i}} \bigg(\,\, }\;\; \sigma_i(h_i^t)(a_i) \sigma_{-i}(h_{-i}^t)(a_{-i}) P_{(a_i,a_{-i})}[s_i,s_{-i}] \bigg)
\end{align*}

\begin{definition}[Nash Equilibrium]
A set of strategies $\sigma = (\sigma_1,\ldots,\sigma_n)$ is a Nash equilibrium of a repeated game if for every player $i$ and every strategy $\sigma_i'$:
$$V_{i, \sigma}(\emptyset) \geq V_{i, (\sigma_{-i}, \sigma_i')}(\emptyset)$$
\end{definition}

A subgame perfect Nash equilibrium requires that from any history, the continuation play forms a Nash equilibrium.
\begin{definition}[Subgame perfect Nash equilibrium]\label{def:spne}
In a game of perfect monitoring, a set of strategies $\sigma = (\sigma_1,\ldots,\sigma_n)$ is a subgame-perfect Nash equilibrium of a repeated game if for every player $i$, every finite history $h^t$ and every strategy $\sigma_i'$:
$$V_{i, \sigma}(h^t) \geq V_{i, (\sigma_{-i}, \sigma_i')}(h^t)$$
\end{definition}

In the case of imperfect monitoring, histories are not common knowledge (agents privately observe the actions they take), there is no well-defined ``subgame.'' There is, however, one in the case where all players use public strategies.
%
%
Formally, in games of public monitoring, we can define a ``perfect public equilibrium.'' This is a particularly structurally simple equilibrium concept in repeated games, but one in which a folk theorem is known to hold. Although our results hold more generally, we will first prove our ``anti-folk-theorem'' for perfect public equilibrium.

\begin{definition} \label{def:ppe}
A \emph{perfect public equilibrium} (PPE) is a profile of public strategies for each player (Definition \ref{def:public_strat}) such that the continuation play from every public history constitutes a Nash Equilibrium.
\end{definition}
\citeasnoun{fudenberg1994folk} show that as long as the signalling structure is rich enough (satisfying a full rank condition), a Folk theorem holds in perfect public equilibria.%

Finally, to consider general (not necessarily public) strategies in games of public monitoring, and games of private monitoring, there are no counterparts to subgame perfect/ perfect public equilibria. Consequently, the refinement that is normally considered is that of sequential equilibrium, which requires that for any player $i$ and  any private history $h_i^t$, the continuation play is a best reply to the play of others given the agent's beliefs about the private histories they observe. Formally:
\begin{definition} \label{def:sequential}
Suppose game is one of no observable deviations (Definition \ref{def:nod}). A strategy profile $\sigma$ is a \emph{sequential equilibrium} if for all private histories $h_i^t$, $\sigma_i(h_i^t) $ is a best reply to $\sigma_{-i}(h^t_{-i})|h_i^t. $
\end{definition}
As we alluded to earlier, the assumption of no observable deviations is made to ensure that the conditional distribution $\sigma_{-i}(h^t_{-i})|h_i^t $ is well defined at every private history.

\section{Signal Privacy Yields Anti-Folk Theorems}

\subsection{Public Monitoring}

\subsubsection{Perfect Public Equilibria}
As we said earlier, the main ideas in our construction are easiest to explain for the case of public perfect equilibria (Definition \ref{def:ppe}). The intuition is essentially encapsulated as follows: in general, the only thing that prevents an agent at some stage $t$ from deviating from his equilibrium strategy, and playing a (stage-game) best response to the distribution over his opponent's actions is fear of punishment: if his opponents can detect this deviation, then they can change their behavior to lower his expected \emph{future} payoff. Differential privacy provides us a simple, worst-case way of quantifying the decrease in expected future payoff that can result from a single player's one-stage unilateral deviation. If this decrease can be made small enough, then it cannot serve as an incentive to prevent any player from playing anything other than an (approximate) stage-game best response to his opponents. Thus, every day, all players must be playing an approximate equilibrium of the stage game. We formalize this intuition in the next theorem:

\begin{theorem} \label{thm:ppe}
Fix any repeated game with discount factor $\delta$, with public signals that satisfy $(\epsilon, \gamma)$-signal privacy. Let $\sigma = (\sigma_1,\ldots,\sigma_n)$ be a perfect public equilibrium (Definition \ref{def:ppe}). Then for every history $h^t$, the distribution on actions that result, $(\sigma_1(h^t),\ldots,\sigma_n(h^t))$, forms an $\eta$-approximate Nash equilibrium of the stage game, for
$$\eta = \frac{\delta}{1-\delta}(\epsilon + \gamma).$$
\end{theorem}
\begin{proof}

Since $\sigma$ forms a perfect public equilibrium of the repeated game, for every public history $h^t$ we have:
$$V_{i, \sigma}(h^t) \geq V_{i, (\sigma_{-i}, \sigma_i')}(h^t)$$
Expanding this definition, and noting that in particular single stage deviations cannot be profitable, we know for every single-stage deviation $a_i'$:
\begin{align*}
&(1-\delta) u_i(\sigma(h^t)) + \delta  \sum_{s \in S} V_{i,\sigma}((h^{t},s)) P_{\sigma(h^t)}(s) \\
\geq& (1-\delta) u_i(a'_i, \sigma_{-i}(h^t)) + \delta  \sum_{s \in S} V_{i,\sigma}((h^{t},s)) P_{(a'_i, \sigma_{-i}(h^t))}(s).
\end{align*}
By the definition of $(\epsilon,\gamma)$-signal privacy, we also know that for all $s$,
\begin{align*}
P_{\sigma(h^t)}(s) \geq \exp(-\epsilon)P_{(a'_i, \sigma_{-i}(h^t))}(s) - \gamma
\end{align*}
Since $V_{i,\cdot}$ is the normalized infinite discounted sum of numbers between $0$ and $1$, $V_{i,\cdot} (\cdot) \in [0,1]$. We can therefore derive:
\begin{align*}
&(1-\delta) u_i(\sigma(h^t))\\
\geq& (1-\delta) u_i(a'_i, \sigma_{-i}(h^t)) + \delta\sum_{s \in S} V_{i,\sigma}((h^{t},s))\left(P_{(a'_i, \sigma_{-i}(h^t))}(s) - P_{\sigma(h^t)}(s)\right)  \\
\geq& (1-\delta) u_i(a'_i, \sigma_{-i}(h^t)) + \delta\sum_{s \in S} V_{i,\sigma}((h^{t},s))\left(\exp(-\epsilon)P_{(a'_i, \sigma_{-i}(h^t))}(s) - \Delta_{s} - P_{\sigma(h^t)}(s)\right)
\end{align*}
where $\Delta_s$ is the smallest non-negative value that satisfies the inequality
\begin{align*}
P_{\sigma(h^t)}(s) \geq \exp(-\epsilon)P_{(a'_i, \sigma_{-i}(h^t))}(s) - \Delta_s.
\end{align*}
Note that for signals $s$ such that $P_{\sigma(h^t)}(s) \leq \exp(-\epsilon) P_{(a'_i, \sigma_{-i}(h^t))}(s)$, $\Delta_s = 0$. Let $S_+ = \{s \in S : \Delta_s > 0$\}.  We continue, noting that $\exp(-\epsilon) > 1-\epsilon$ for $\epsilon < 1$:
\begin{align*}
&(1-\delta) u_i(\sigma(h^t)) \\
\geq& (1-\delta) u_i(a'_i, \sigma_{-i}(h^t)) + \delta\sum_{s \in S} V_{i,\sigma}((h^{t},s))\left(- \epsilon  P_{\sigma(h^t)}(s) - \Delta_s\right) \\
\geq& (1-\delta) u_i(a'_i, \sigma_{-i}(h^t)) - \delta \left(\epsilon + \sum_{s \in S_+} V_{i, \sigma}((h^{t},s))\Delta_s \right) \\
\geq& (1-\delta) u_i(a'_i, \sigma_{-i}(h^t)) - \delta(\epsilon +\gamma).
\end{align*}
The last inequality follows because $\sum_{s \in S_+} \Delta_s$ is the smallest value $\Delta$ that satisfies the inequality $P_{\sigma_i(h^t)}(S_+)  \geq \exp(-\epsilon)P_{(a'_i, \sigma_{-i}(h^t))}(S_+) - \Delta$, but by our guarantee of $(\epsilon, \gamma)$ signal privacy, we know that we must have $\Delta \leq \gamma$, and so $\sum_{s \in S_+} \Delta_s \leq \gamma$.
Dividing both sides by $1-\delta$ we get:
$$u_i(\sigma(h^t)) \geq u_i(a'_i, \sigma_{-i}(h^t)) - \frac{\delta}{1-\delta}(\epsilon+\gamma)$$

In other words, at any history $h^t$, the prescribed strategy profile $\sigma(h^t)$ must form a $\frac{\delta}{1-\delta} (\epsilon+\gamma)$-approximate Nash equilibrium of the stage game.
\end{proof}

With a bit more care, we can prove a related theorem that extends to the case of non-public strategies. The proof is similar, but with two caveats:

First, we must now consider a deviation of player $i$, $\sigma_i'$ that not only makes a single stage game deviation to $a_i'$ at some day $t$, but then continues to play using an artificial history that for every day $t' > t$, records that on day $t$, player $i$ played an action drawn from $\sigma_i(h_i^t)$ instead of having played $a_i'$.

Second, the strategies that form the stage-game equilibrium are not exactly $\sigma_i(h^t_i)$, for with these strategies, each player might be best-responding not to the strategies of his opponent, but to the \emph{expected} strategies of his opponents (defined by their unknown histories $h^t_j$, which take into account not just the publicly observed signals, but also the privately known action history), conditioned on the public part of the history.

Instead, the strategies that form the stage-game equilibrium are the randomized strategies $\sigma_i(h^t_i | s)$, where the randomness is taken over the realization of the players (privately known) action history, conditioned on the (publicly observed) signal history. We begin with some notation: Given a set of $t-1$ public signals $s^t = (s_1,\ldots,s_{t-1})$ and a vector of actions for player $i$, $\ba_i^t = (a_i^1,\ldots,a_i^{t-1})$, write $h_i^t(s^t, \ba_i^t) = (a_i^1, s_1, \ldots, a_i^{t-1}, s_{t-1})$ to denote the history that combines them. Given a strategy $\sigma_i$ for player $i$ mapping histories to actions, together with a vector of public signals $s^t$ write: $\hat{\sigma}_{i | s^t}$ to denote the probability distribution that plays each action $a_i$ with probability equal to:
$$\Pr[a_i \sim \hat{\sigma}_{i | s^t}]= \sum_{\ba_i^t \in A_i^{t-1}} \Pr[\ba_i^t | s^t] \cdot \sigma_i(h_i^t(s^t, \ba_i^t))$$
In other words, $\hat{\sigma}_{i | s^t}$ represents the distribution defined over actions played by player $i$, $\sigma_i(h_i^t)$ when all that is known is the public history, and the randomness is both over the choice of actions $\ba_i^t$ that define $h_i^t$, as well as the randomness in the mixed strategies $\sigma_i(h_i^t)$.

We can now state our theorem:
\begin{theorem} \label{eqn:seq_public}
Fix any repeated game with discount factor $\delta$, with public signals that satisfy $(\epsilon, \gamma)$-signal privacy and no observable deviation (Definition \ref{def:nod}). Let $\sigma = (\sigma_1,\ldots,\sigma_n)$ be a sequential equilibrium of the repeated game (Definition \ref{def:sequential}). Then for every public history of signals $s^t$, the distribution on actions played on day $t$, $(\hat{\sigma}_{1 | s^t},\ldots,\hat{\sigma}_{n | s^t})$, forms an $\eta$-approximate Nash equilibrium of the stage game, for
$$\eta = \frac{\delta}{1-\delta}(\epsilon + \gamma).$$
\end{theorem}

The proof is deferred to the appendix.

\subsubsection{Non-subgame perfect equilibria}

We can prove a similar theorem even for non-subgame perfect equilibrium -- except that now our claim only applies to histories that occur with nonzero probability when players play according to the specified equilibrium.
\begin{theorem}
Fix any repeated game with discount factor $\delta$, with public signals that satisfy $(\epsilon, \gamma)$-differential privacy. Let $\sigma = (\sigma_1,\ldots,\sigma_n)$ denote a public equilibrium (not necessarily subgame perfect). Then for every history $h^t$ that occurs with positive probability when players play according to $\sigma$, the distribution on actions at stage $t$, $(\sigma_1(h^t),\ldots,\sigma_n(h^t))$ forms an $\eta$-approximate Nash equilibrium of the stage game, for
$$\eta = \frac{\delta}{1-\delta}(\epsilon + \gamma)$$
\end{theorem}
\begin{proof}
We will write $\Pr_\sigma[h^t]$ to denote the probability that a given public history $h^t$ arises when players use strategies $\sigma$. For each $j \leq t$, write $h^{t, \leq j}$ to denote the sub-history of $h^t$ consisting of the first $j$ periods, and $h^{t,j}$ to be the $j^{\text{th}}$ period history. Then:
$$\Pr_\sigma[h^t] = \prod_{j=1}^t P_{\sigma(h^{t,\leq j})}(h^{t,j})$$
Now fix any $T.$ We can write:
$$V_{i, \sigma}(\emptyset) = (1-\delta)\left(\sum_{h^{T} \in S^{T}}\Pr_\sigma[h^{T}]\left(\left(\sum_{t=0}^{T}\delta^tu_i(\sigma(h^{T,\leq t}))+\frac{\delta^{T+1}}{1-\delta}\sum_{s \in S}V_{i, \sigma}(h^{T}, s)P_{\sigma(h^t)}(s)\right)\right)\right)$$

Consider any history $h^T$ such that $\Pr_\sigma[h^{T}] > 0$, and consider the deviation of player $i$, $\sigma_i'$ that is identical to $\sigma_i$, except that on history $h^T$ player $i$ plays a stage-game best response to his opponents. $\sigma_i'(h^T) = \arg\max_{a \in A_i}u_i(a, \sigma_{-i}(h^T)) \equiv a_i^*$. Since $\sigma$ is a Nash equilibrium of the repeated game, we know that $V_{i, \sigma}(\emptyset) - V_{i, (\sigma_i', \sigma_{-i})}(\emptyset) \geq 0$. Since $\Pr_\sigma[h^{T}] > 0$, we can divide and write this difference as:
\begin{eqnarray*}
0 &\leq& \frac{1}{(1-\delta)\Pr_{\sigma}[h^T]}\left(V_{i, \sigma}(\emptyset) - V_{i, (\sigma_i', \sigma_{-i})}(\emptyset)\right) \\
 &=& \left(\delta^T\left(u_i(\sigma(h^T)) - u_i(a_i^*, \sigma_{-i}(h^T))\right) + \frac{\delta^{T+1}}{1-\delta}\sum_{s \in S}V_{i, \sigma}(h^T, s)\left(P_{\sigma(h^t)}(s) - P_{(a_i^*, \sigma(h^T)})(s)\right)\right) \\
 &\leq& \left(\delta^T\left(u_i(\sigma(h^T)) - u_i(a_i^*, \sigma_{-i}(h^T))\right) + \frac{\delta^{T+1}}{1-\delta}\sum_{s \in S}\left(P_{\sigma(h^t)}(s) - P_{(a_i^*, \sigma(h^T))}(s))\right)\right) \\
 &\leq& \left(\delta^T\left(u_i(\sigma(h^T)) - u_i(a_i^*, \sigma_{-i}(h^T))\right) + \frac{\delta^{T+1}}{1-\delta}(\epsilon + \gamma)\right)
\end{eqnarray*}
where the last inequality follows from $(\epsilon,\gamma)$-differential privacy, and the fact that $\exp(-\epsilon) \geq 1-\epsilon$.
Dividing through by $\delta^T$ and rearranging, we find:
$$\left(u_i(a_i^*, \sigma_{-i}(h^T))-u_i(\sigma(h^T))\right) \leq \frac{\delta}{1-\delta}\left(\epsilon + \gamma\right)$$
which completes the proof.
\end{proof}

\subsection{Private Monitoring}
Next, we consider the case of private monitoring. Our theorem here is slightly weaker: that when playing a game with differentially private signals, every equilibrium of the repeated game must at each stage play an approximate \emph{correlated} equilibrium of the stage game. The reason is natural: agents are no longer aware of the history that their opponents are viewing, and so can no longer consider deviations that are best responses to the distributions being played on that day by their opponents. However, they may at any stage consider deviations that base their action on the past private signals that they observe, which gives a posterior distribution on the signals that their opponents observe. Since the distribution on signals that players receive can be correlated; the appropriate stage game solution concept is correlated equilibrium.
\begin{theorem} \label{thm:private}
Fix any repeated game with discount factor $\delta$ and private signals that satisfy $(\epsilon, \gamma)$-differential privacy and no observable deviations (Definition \ref{def:nod}). Let $\sigma = (\sigma_1,\ldots,\sigma_n)$ denote a sequential equilibrium of the repeated game (Definition \ref{def:sequential}). Then for every history $\bh^t \equiv (h_1^t, h_2^t, \ldots, h_n^t)$ that occurs with positive probability when players play according to $\sigma$, the distribution on actions at stage $(\sigma_1(h_1^t),\ldots,\sigma_n(h_n^t))$ forms an $\eta$-approximate correlated equilibrium for
\begin{align*}
\eta = \frac{\delta}{1-\delta} (\epsilon + \gamma).
\end{align*}
\end{theorem}

\begin{proof}
Fix a sequential equilibrium $\sigma$. Then, for every player $i$ and every history $h_i^t$, we can define a posterior distribution $\mathcal{D}(i, h_i^t, \sigma_{-i})$ on histories $h_{-i}^t$ observed by the other players. In other words the probability that the other players observed private histories $h_{-i}^t$ given player $i$'s private history is $h_i^t$ and they play strategies $\sigma_{-i}$ is given as $\mathcal{D}(i, h_i^t, \sigma_{-i})(h_{-i}^t)$.

Fix a player $i$ and a period $t$. Consider a private history $h_i^t$ that can arise with positive probability when all players play according to $\sigma$ from periods $1$ to $t-1$. Define
$$a^*_i \in \arg\max_{a_i \in A_i} \sum_{h_{-i}^t}u_i(a_i, \sigma_{-i}(h_{-i}^t)) \mathcal{D}(i, h_i^t, \sigma_{-i})(h_{-i}^t).$$
In words $a^\star_i$ is the period $t$ best response for player $i$ who observes private history $h_i^t$.

Define $\sigma'_i$ to be the strategy that is identical to $\sigma_i$, except on history $h^t_i$, it plays $\sigma_i'(h^t_i) = a_i^*$, and then play in future periods is according to $\sigma_i$  as if an action drawn from $\sigma_i(h^t_i)$ was played in period $t$. Formally, for any period $\tau> t$, with realized history $h_i^\tau$, the deviation involves playing $\sigma_i(h_i^{\prime \tau})$, where $h_i^{\prime \tau}$ is the same as $h_i^\tau$ except in the component corresponding to the action played at time $t$; i.e. $a_i^{\prime t} \sim \sigma_i(h_t^i)$ whereas $a_i^t = a^*$ by definition.

Further, given strategies $\sigma$, and any history $\bh^t = (h_1^t,\ldots,h_n^t)$ for the $n$ players, define $W_{i,\sigma}(\bh^t)$ as:
\begin{align*}
W_{i,\sigma}(\bh^t) =& (1-\delta) u_i(\sigma_i(h_i^t), \sigma_{-i}(h_{-i}^t)) \\
&+ \delta \sum_{s_i,s_{-i},a_i, a_{-i}} W_{i,\sigma}((\bh^t, (s_i,s_{-i}), (a_i,a_{-i}))) \sigma_i(h_i^t)(a_i) \sigma_{-i}(h_{-i}^t)(a_{-i}) P_{(a_i,a_{-i})}[s_i,s_{-i}]
\end{align*}
Note that the argument to $W_i$ is the vector of private histories of all agents, $\bh^t$. It represents the expected discounted value to agent $i$ continuing from a point where his private history is $h_i^t$ \emph{and} the private histories of other players are $h_{-i}^t$. Of course player $i$ only knows his own private history, therefore player $i$'s continuation value playing $\sigma_i$ after a private history $h_i^t$ is
\begin{align*}
V_{i,\sigma}(h_i^t) =& \sum_{h_{-i}^t} W_{i,\sigma}((h_i^t,h_{-i}^t))  \mathcal{D}(i, h_i^t, \sigma_{-i})(h_{-i}^t),\\
\implies V_{i,\sigma}(h_i^t)=& (1- \delta) \sum_{h_{-i}^t}u_i(\sigma_i(h_i^t), \sigma_{-i}(h_{-i}^t)) \mathcal{D}(i, h_i^t, \sigma_{-i})(h_{-i}^t)\\
&+ \delta \sum_{h_{-i}^t,s_i,s_{-i},a_i, a_{-i}} \bigg( W_{i,\sigma}(((h_i^t,h_{-i}^t), (s_i,s_{-i}), (a_i,a_{-i})))\mathcal{D}(i,h_i^t,\sigma_{-i})(h_{-i}^t) \\
& \hphantom{+ \delta \sum_{h_{-i}^t,s_i, a_i, s_{-i}, a_{-i}} \bigg(\,\, }\;\; \sigma_i(h_i^t)(a_i) \sigma_{-i}(h_{-i}^t)(a_{-i}) P_{(a_i,a_{-i})}[s_i,s_{-i}] \bigg).
\end{align*}
Now consider $\sigma_i'$ to be the deviation we described above. Note that by playing $\sigma_i'$ from history $h_i^t$ onwards, player $i$'s expected discounted payoff can be written as:
\begin{align*}
V_{i, (\sigma_{-i}, \sigma_i')}(h_i^t) = &(1-\delta) \sum_{h_{-i}^t}u_i(a_i^*, \sigma_{-i}(h_{-i}^t)) \mathcal{D}(i, h_i^t, \sigma_{-i})(h_{-i}^t)\\
& + \delta \sum_{h_{-i}^t,s_i, s_{-i},a_i,  a_{-i}} \bigg( W_{i,\sigma}(((h_i^t,h_{-i}^t),(s_i,s_{-i}),(a_i,a_{-i})))  \mathcal{D}(i,h_i^t,\sigma_{-i})(h_{-i}^t) \\
& \hphantom{+ \delta \sum_{h_{-i}^t,s_i,  s_{-i}, a_i, a_{-i}} \bigg(\,\, }\;\;\sigma_i(h_i^t)(a_i) \sigma_{-i}(h_{-i}^t)(a_{-i}) P_{(a^*_i,a_{-i})}[s_i,s_{-i}] \bigg)
\end{align*}
Since $\sigma$ forms a sequential equilibrium of the repeated game, for every history $h_i^t$ that occurs with positive probability, we have:
$$V_{i, \sigma}(h_i^t) \geq V_{i, (\sigma_{-i}, \sigma_i')}(h_i^t).$$
Substituting the definitions of $V_{i, \sigma}(h_i^t)$ and $V_{i, (\sigma_{-i}, \sigma_i')}(h_i^t)$ into this inequality, we have:
\begin{align*}
&  \sum_{h_{-i}^t}\left (u_i(a^*, \sigma_{-i}(h_{-i}^t)) -  u_i(\sigma_i(h_i^t), \sigma_{-i}(h_{-i}^t))   \right)\mathcal{D}(i, h_i^t, \sigma_{-i})(h_{-i}^t) \\
\leq  & \frac{\delta}{1-\delta} \sum_{h_{-i}^t,s_i,  s_{-i},a_i, a_{-i}} \bigg( W_{i,\sigma}(((h_i^t,h_{-i}^t),(s_i,s_{-i}),(a_i,a_{-i})))  \mathcal{D}(i,h_i^t,\sigma_{-i})(h_{-i}^t) \\
& \hphantom{+ \delta \sum_{h_{-i}^t,s_i, s_{-i}, a_i, a_{-i}} \bigg(\,\,\; }\;\;\;\sigma_i(h_i^t)(a_i) \sigma_{-i}(h_{-i}^t)(a_{-i}) \left(P_{a_i,a_{-i}}[s_i,s_{-i}] - P_{(a^*_i,a_{-i})}[s_i,s_{-i}]\right) \bigg)
\end{align*}
By $(\epsilon,\gamma)$-differential privacy of the private signal $P_{\ba}(\mathbf{s})$, and the fact that $\exp(-\epsilon) \geq 1-\epsilon$, therefore, we have
\begin{align*}
\sum_{h_{-i}^t}\left (u_i(a_i^*, \sigma_{-i}(h_{-i}^t)) -  u_i(\sigma_i(h_i^t), \sigma_{-i}(h_{-i}^t))   \right)\mathcal{D}(i, h_i^t, \sigma_{-i})(h_{-i}^t)  \leq & \frac{\delta}{1-\delta} (\epsilon + \gamma).
\end{align*}
By Definition \ref{def:ce}, therefore, the strategies $\sigma_i$ involve the play of $\frac{\delta}{1-\delta} (\epsilon + \gamma)$-approximate correlated equilibrium in every period: the analogue of the private signal observed by players in that definition is the private history $h_i^t$, which are correlated via the strategies $\sigma_i$ and monitoring structure.
\end{proof}

\section{Examples of Games with Signal Privacy} \label{sec:examples}

The rest of this paper will be devoted to showing by example that such differentially private signal structures may naturally arise in large games. In particular, we will show that if the number of players $n$ is large, $\epsilon$ and $\gamma$ may naturally be small. This suggests that in large societies, inter-temporal transfers may not support ``much more'' than repeated play of the static equilibria of the stage game.

Gaussian (i.e. normally distributed) mean-zero noise is a natural form of perturbation that arises in many settings of impartial information (it is the limiting distribution of many independent random processes). It also happens to be a noise distribution that guarantees differential privacy. Formally, a theorem from \citeasnoun{DKMMN06} will be useful.%
\footnote{A proof of this theorem in the form stated here can be found in Appendix A of \citeasnoun{HR12}.}
First we define the ($L^2$) sensitivity of a function.
\begin{definition}
A function $f: T^n \rightarrow \Re^d$ has $L^2$ sensitivity $s$ if for any $\mathbf{t} \in T^n$, $ i \leq n$ and $t' \in T$, we have that
\begin{align*}
\lVert f(\mathbf{t}) - f (t', t_{-i})   \rVert_2 \leq s,
\end{align*}
where $\lVert \cdot \rVert_2$ is the standard Euclidean norm in $\Re^d$.
\end{definition}

\begin{theorem}[\citeasnoun{DKMMN06}]\label{thm:DKMMN06} Suppose a function $f:T^n \rightarrow \Re^d$ has $L^2 $ sensitivity $s$. Then the algorithm that computes $f(\mathbf{t}) + Z$ is $(\epsilon, \gamma)$-differentially private, where $Z \in \Re^d$ is a random vector where each coordinate is drawn i.i.d. from the normal distribution $N(0, \sigma^2)$ for
\begin{align*}
\sigma = \frac{s}{\epsilon}  \sqrt{\log\left(\frac{1.25}{\gamma} \right)}.
\end{align*}
\end{theorem}

The following corollary will be useful in our analysis.
\begin{corollary} \label{cor:sensitive}
Suppose a function $f:T^n \rightarrow \Re^d$ has $L^2 $ sensitivity $s(n)$ (here we assume the sensitivity is a function of $n$, which should be thought of as diminishing in $n$). Then for any constant $\sigma$, the algorithm that computes $f(\mathbf{t}) + Z$ where $Z \in \Re^d$ is a random vector where each  is drawn i.i.d. from the normal distribution $N(0, \sigma^2)$ is $(\epsilon, \gamma)$-differentially private for
\begin{align*}
\epsilon + \gamma \in O\left(s(n) \sqrt{\log \tfrac{1}{s(n)} } \right)
\end{align*}
In particular if the sensitivity $s(n) \in O(\frac1n)$, we have $\epsilon + \gamma \in O(\frac{\sqrt{\log n}}{n})$
\end{corollary}

\subsection{Large Anonymous Games}
A class of games that has received much attention in the literature is the class of Large Anonymous games. In these games, all players have the same set of available actions and each player's payoff depends only on the action he takes and the histogram of actions taken by others (and not on the identity of the players).%
\footnote{Papers in this literature often refer to such games simply as Large games.}

Consider a repeated game that has such a large anonymous game as a stage game. A player's payoff from playing action $a \in A$, when the population distribution of plays by others is $\alpha \in \Delta A$ is $u_i(a,\alpha) \in [0,1]$. Suppose further that in every period, a noisy histogram of the actions taken by all agents is announced, i.e. when the realized distribution of plays is $\alpha \in \Delta A$, the announced distribution is $\alpha + Z$ for $Z \in \Re^{|A|}$, with each component of $Z \sim N(0,\sigma^2)$. Then, fixing players' discount factor $\delta$, by Corollary \ref{cor:sensitive} and Theorem \ref{thm:ppe} that any perfect public equilibrium of the repeated game must involve each player playing an $\eta-$approximate Nash equilibrium of the stage game for $$\eta \in O\left(\frac{\delta}{1-\delta} \frac{\sqrt{\log n}}{n} \right).$$
Note that as $n$ tends large, $\eta \rightarrow 0$.

Note that a similar noise distribution results if $\alpha \in \Delta A$ is computed by \emph{subsampling} player actions. It can similarly be shown that the distribution that results from subsampling (sufficiently to get a constant error rate in the action histogram) results in an $(\epsilon,\gamma)$ differential privacy guarantee with $(\epsilon+\gamma)$ tending to $0$ as $n$ grows large. 

\subsection{Noisy Cournot Games}
One of the classic empirical applications of repeated games is to study the possibility of collusion in repeated oligopolistic competition---see \citeasnoun{green1984noncooperative} for a classic reference. We demonstrate via their model that in a large oligopoly, collusion may be impossible.

The stage game is $n$-firm Cournot competition. Each firm $i$ simultaneously chooses quantity $q_i \in [0,1]$ to produce, at a cost of $c(q_i)$. For simplicity in calculations, take $c(q_i) = q_i$. After all firms have selected quantities, a price is determined from the (realized, stochastic) demand as
\begin{align*}
p = \theta P\left(\frac{1}{n} \sum_{i=1}^n q_i \right),
\end{align*}
where $P(\cdot)$ is a continuously differentiable decreasing demand function and $\theta$ is the realized demand shock. We assume that $\theta$ is distributed log-normally with $\mathbb{E}[\theta]=1$. Each firm $i$ then realizes a profit of its revenue less cost, i.e., $pq_i - c(q_i)$. Further, only this price $p$ is observed publicly by all the firms---firms do not directly observe the quantities produced by other firms. Taking logarithms, we have that
\begin{align*}
\log p = \log \theta + \log P\left(\frac{1}{n} \sum_{i=1}^n q_i \right).
\end{align*}

For $n$ large, by a Taylor series expansion, we have that the sensitivity of $\log P(\cdot)$  is $$\frac1n \sup_{x \in [0,1]}\frac{P'(x)}{P(x)}.$$
Note that fixing a function $P$, this quantity is diminishing at a rate of $O(1/n)$. Note also that since $\theta$ is log-normally distributed, $\log \theta$ is normally distributed. Therefore by Corollary \ref{cor:sensitive} and Theorem \ref{thm:ppe}, any Perfect Public equilibrium of the stage game must involve the play of approximate Nash Equilibrium of the stage game in every period for $\eta \in O\left( \frac{\delta}{1-\delta} \sqrt{\frac{\log n}{n}}\right)$.

\subsection{Noisy Counterfactual Payoffs}
Now let us consider an example of an imperfect private monitoring setting. Fix a $n$ player game, and assume that the game is $\mu$--sensitive, i.e. for every player $i$ and every player $j \neq i$, $j$ changing the action he plays can affect player $i$'s payoff by at most $\mu$ \emph{regardless} of what everyone else is playing:
\begin{align*}
\forall i \neq j, a_i, a_j, a_j', a_{-ij}: \quad \lvert u_i(a_i, a_j, a_{-ij}) - u_i(a_i, a'_j, a_{-ij} )\rvert \leq \mu.
\end{align*}
After each period, each player receives a private signal of a (noisy) estimate of the payoff she would have received in the past period for each of the actions she could have taken, i.e. $S_i =\Re^{|A_i|}$, and when other players play $a_{-i}$, player $i$ receives a signal where the component corresponding to action $a$ is $u_i(a, a_{-i}) + N(0,\sigma^2).$ Such a signal structure has been widely studied in both the game theory and machine learning literatures. In the game theory literature, \citeasnoun{hart2000simple} and \citeasnoun{foster1997calibrated} show that simple heuristic strategies (minimizing regret) played in this setting will result in the empirical distribution of play converging to the set of correlated equilibria, see e.g. Chapter $4$ of \citeasnoun{nisan2007algorithmic} for a survey and the connections to machine learning.

We note that in a $k$-action game, the resulting signal distribution is over a vector of length $k\cdot n$ (reporting the payoff to each of the $n$ players for each of their $k$ actions). Since each action is individually $\mu$ sensitive, the $\ell_2$ sensitivity of the entire vector of payoffs is $O(\sqrt{nk}\mu)$. Fixing players' discount factor $\delta$, by Corollary \ref{cor:sensitive} and Theorem \ref{thm:ppe} it follows that any perfect public equilibrium of the repeated game must involve each player playing an $\eta-$approximate Correlated equilibrium of the stage game for
$$\eta \in O\left(\frac{\delta}{1-\delta} \mu \sqrt{n\cdot k\log \frac{1}{\mu}} \right).$$
Note that as $n$ tends large, if $\mu \ll \frac{1}{\sqrt{n}}$, $\eta \rightarrow 0$.

\subsection{Other Settings with Private Signaling} \label{sec:otherres}
Finally, we remark that the three settings we have noted in which the signals are naturally differentially private (noisy anonymous games, noisy Cournot games, and noisy counterfactual payoffs) are not isolated examples. Differential privacy is a powerful measure of ``influence'' in part because it enjoys a strong ``composition'' property. A sequence of $(\epsilon, \gamma)$-differentially private distrubutions (even chosen adaptively, as a function of the realized values of previous distributions) remain $(\epsilon', \gamma')$-differentially private, for values of $\epsilon'$ and $\gamma'$ depending on the number of compositions and the parameters $\epsilon$ and $\gamma$. Informally, the $\gamma$'s ``add up'' linearly, whereas the $\epsilon$ parameters increase only with the square root of the number of such compositions.\footnote{It is this property that makes differential privacy a more appealing notion of influence than mere statistical distance (which corresponds to $(0, \gamma)$-differential privacy). If we were to talk only about statistical distance, then the degradation of the statistical distance parameter would always be linear in the number of compositions, whereas the degradation in the $\epsilon$ parameter in differential privacy can be much slower.} Formally, we have the following theorem due to \citeasnoun{DRV10}:

\begin{theorem}[Composition for differential privacy, \citeasnoun{DRV10}]
Let $M_1(D),\ldots , M_k(D)$ be families of distributions such that each $M_i$ is $(\epsilon_i,\gamma_i)$-differentially private with
$\epsilon_i \leq \epsilon'$.  Then for any $\gamma \in [0,1]$: $M(D) = (M_1(D), \ldots , M_k(D))$ is
$(\epsilon,\gamma + \sum_{i=1}^k \gamma_i)$-private for
\[
  \epsilon = \sqrt{2 \log(1 / \gamma) k}\epsilon' + k\epsilon'(e^{\epsilon'} - 1)
\]
$M(D)$ also satisfies $(\epsilon,\gamma)$ differential privacy for $\epsilon = \sum_{i=1}^k \epsilon_i$ and $\gamma = \sum_{i=1}^k \gamma_i$.\footnote{For large values of $k$, this simpler statement is substantially weaker than the more sophisticated composition statement. However, it can be stronger for small values of $k$.} Moreover, these compositional properties hold even if each distribution $M_i$ is chosen adaptively, i.e. only after observing the realized values of distributions $M_1(D),\ldots,M_{i-1}(D)$.
\end{theorem}

Most immediately, this means that our results hold for any combination of the signalling schemes discussed so far: for example, agents could observe \emph{both} a noisy price, \emph{and} a noisy histogram of actions that were played, and the resulting composite signal would remain differentially private. More generally, a large number of simple noisy ``primitives'' guarantee differential privacy: we have already seen that the natural operation of adding Gaussian noise to a low sensitivity vector is differentially private. However, many other distributions guarantee differential privacy as well, including the ``Laplace'' distribution (which is a symmetric exponential distribution, and gives an even stronger guarantee of $(\epsilon, 0)$-differential privacy compared to Gaussian noise) \cite{DMNS06}. Other simple operations such as subsampling $k$ out of $n$ player actions guarantee $(0, k/n)$-differential privacy.

Moreover, differential privacy does not degrade with post-processing. The following fact is a simple consequence of the definition:
\begin{theorem}[Post-processing]
Let $M(D)$ be an $(\epsilon,\gamma)$-differentially private family of distributions over some range $\mathcal{O}$, and let $f:\mathcal{O}\rightarrow \mathcal{O}'$ be any randomized mapping. Then $f(M(D))$ is an $(\epsilon,\gamma)$-differentially private family of distributions over $\mathcal{O}'$.
\end{theorem}
What this means is that any signal distribution that can be described as a differentially private distribution (resulting, from e.g. the composition theorem above), followed by post-processing (for example, we could have a noisy estimate of the fraction of players playing each action, \emph{truncated and renormalized} to form a distribution) remains differentially private.

As a result of the composition theorem and the post-processing theorem, differentially private signaling distributions are not difficult to come by in natural settings. Most noisy settings can be seen to satisfy $(\epsilon,\gamma)$-privacy for some finite values of $\epsilon$ and $\gamma$. Our theorems apply when these parameters $\epsilon$ and $\gamma$ tend to zero as $n$ grows large -- and as we have seen, this is generally the case when the noise rate is fixed (independent of $n$), and the game is large.

\section{Conclusions and Open Questions}

We have shown that in large games with (public or private) imperfect monitoring, it can naturally arise that the signaling structure satisfies the constraint of ``differential privacy'', in which the privacy parameters $\epsilon$ and $\gamma$ tend to 0 as the size $n$ of the game grows large. Moreover, we have shown that whenever this is the case the equilibria of the repeated game collapse to the set of equilibria of the stage game in large populations.  This conclusion holds for a broad class of equilibria of the repeated game, that e.g.~need not necessarily even be subgame perfect.  With public monitoring, the equilibrium set collapses to the set of \emph{Nash} equilibria of the stage game, and with private monitoring, the equilibrium set collapses to the set of \emph{correlated} equilibria of the stage game.

This ``anti-folk-theorem'' for large, noisy repeated games has several interpretations. On the one hand, it means that cooperative, high welfare folk-theorem equilibria which cannot be supported as equilibria of the stage game are also impossible to support in repeated games from this natural class. On the other hand, it also eliminates low-welfare equilibria of the repeated game that are not supported in the stage game, and so in general, \emph{improves} the \emph{price of anarchy} when computed over the set of Nash equilibria of the repeated game. Our theorem can also be seen as an argument that equilibria of the repeated game can be more predictive in large noisy games, than in general repeated games. By removing the bite of the folk theorem, the multiplicity of possible equilibria of the repeated game can be greatly reduced, in many cases to a single (or small number of) Nash equilibria of the stage game.

There remain many interesting questions along this line of inquiry. For example, is there an even more general condition than differential privacy that can yield the same results? 
Similarly, there is the question of whether our theorems are qualitatively tight. Certainly they are in public monitoring settings -- it is always possible to support a sequence of stage game Nash equilibria as an equilibrium of the repeated game, and we have shown that this is essentially the \emph{only} class of equilibria that can be supported in the games that we study. However, in the private monitoring setting, we have shown only that equilibria of the repeated game must consist of a series of (approximate) correlated equilibria of the stage game. But is it the case that for every stage game and every sequence of correlated equilibria of that stage game, there exists a private signaling structure that does indeed support the sequence of correlated equilibria as an equilibrium of the repeated game?

More broadly, analysis of ``large'' dynamic and repeated games is easier because each player's action is close to myopic/ non-strategic. Several different notions of largeness have been considered in the literature. For example, \citeasnoun{al2001large} use concentration inequalities to show that only a small number of ``pivotal'' players will play non-myopically/ strategically in any large game (see also more recent results by \citeasnoun{gradwohl2008fault}). In an influential paper, \citeasnoun{mclean2002informational} define a measure of they term ``informational size'' which bounds the amount of information any single player has about a common variable of interest (in an $L^1$ sense), conditional on knowing the information of all other players. Differential privacy has promise as an alternate measure of ``largeness.'' Further the underlying literature studying properties of differentially private processing (recall, e.g., the results surveyed in Section \ref{sec:otherres}) allow for easy use in applications of interest.

\bibliographystyle{econometrica}
\bibliography{repeated}
\appendix
\section*{APPENDIX}
\setcounter{section}{1}
\begin{proof}[\textsc{Proof of Theorem  \ref{eqn:seq_public}}]
Fix a sequential equilibrium $\sigma = (\sigma_1,\ldots,\sigma_n)$. For every player $i$ and every history $h_i^t = h_i^t(a_i^t, s^t)$, we can define a posterior distribution $\mathcal{D}(i, h_i^t, \sigma_{-i})$ on histories $h_{-i}^t$ observed by the other players. We note several things: first, since the public signals $s^t$ are fixed, $\mathcal{D}(i, h_i^t, \sigma_{-i})$ defines a distribution only over the actions $a_{j}^t$ played by each of the other players $j \neq i$. Second, since the actions $a_j^t$ played by the other players are a function only of their own past actions and the public signals $s^t$, they are conditionally independent of actions $a_i^t$ fixing the public history $s^t$. Hence, we can write $\mathcal{D}(i, h_i^t, \sigma_{-i}) \equiv \mathcal{D}(i, s^t, \sigma_{-i})$. Finally, equivalently note that the distribution over actions $a_j^t, a_k^t$ of any pair of players $j \neq k \neq i$ are conditionally independent fixing the public signals $s^t$, and so we can decompose $\mathcal{D}(i, s^t, \sigma_{-i})$ as a product distribution:
$$\mathcal{D}(i, s^t, \sigma_{-i}) = \prod_{j \neq i} \mathcal{D}_j(i, s^t, \sigma_{-i})$$
where $\mathcal{D}_j(i, s^t, \sigma_{-i})$ represents the marginal distribution on histories $h_j^t$ of player $j$ fixing $s^t$.

Fix any set of public signals $s^t$, and private histories $(h_1^t,\ldots,h_n^t)$ for the $n$ players consistent with $s^t$. For any player $i$, define
\begin{align*}
a^* &= \arg\max_{a_i \in A_i} \mathbb{E}_{h_{-i}^t \sim \prod_{j \neq i} \mathcal{D}_j(i, s^t, \sigma_{-i})}[u_i(a_i, \sigma_{-i}(h_{-i}^t))]\\
& = \arg\max_{a_i \in A_i} \mathbb{E}_{a_{-i} \sim (\hat{\sigma}_{1 | s^t},\ldots,\hat{\sigma}_{n | s^t})}[u_i(a_i, a_{-i})]
\end{align*}
 Define $\sigma'_i$ to be the strategy that is identical to $\sigma_i$, except on history $h^t_i$, it plays $\sigma_i'(h^t_i) = a_i^*$, and then play in future periods is as if an action drawn from $\sigma_i(h^t_i)$ was played in period $t$.

Formally, for any period $\tau> t$, with realized history $h_i^\tau$, the deviation involves playing $\sigma_i(h_i^{\tau'})$, where $h_i^{\tau'}$ is the same as $h_i^\tau$ except in the component corresponding to the action played at time $t$; i.e. $a_i^{t'} \sim \sigma_i(h_t^i)$ whereas $a_i^t = a^*$ by definition.
Given the equilibrium strategies $\sigma$, we can write: %
\begin{align*}
 V_{i, \sigma}(h_i^t) =  &(1-\delta) \mathbb{E}_{a_{-i} \sim (\hat{\sigma}_{1 | s^t},\ldots,\hat{\sigma}_{n | s^t})}[u_i(\sigma_i(h_i^t), a_{-i})]\\
&+ \delta \sum_{s \in S, a_i \in A_i}V_{i,\sigma}((h_i^t,s,a_i)) P[s,a_i | h_i^{t}]
 \end{align*}
Now consider $\sigma_i'$ to be the deviation we described above. Note that by playing $\sigma_i'$ from history $h_i^t$ onwards, player $i$'s expected discounted payoff can be written as:
\begin{align*}
V_{i, (\sigma_{-i}, \sigma_i')}(h^t_i) = &(1-\delta) \mathbb{E}_{a_{-i} \sim (\hat{\sigma}_{1 | s^t},\ldots,\hat{\sigma}_{n | s^t})}[u_i(a^*, a_{-i})]\\
& + \delta \sum_{s \in S, a_i \in A_i} V_{i,\sigma}((h_i^{t},s,a_i)) P[s| h_i^{t}, a^*] P[a_i| \sigma(h_i^t)]
\end{align*}
Since $\sigma$ forms a sequential equilibrium of the repeated game, for every history $h_i^t$ we have:
$$V_{i, \sigma}(h_i^t) \geq V_{i, (\sigma_{-i}, \sigma_i')}(h_i^t)$$
Substituting the definitions of $V_{i, \sigma}(h_i^t)$ and $V_{i, (\sigma_{-i}, \sigma_i')}(h_i^t)$ into this inequality, we have:
\begin{align*}
& \left (\mathbb{E}_{a_{-i} \sim (\hat{\sigma}_{1 | s^t},\ldots,\hat{\sigma}_{n | s^t})}[u_i(a^*, a_{-i})] -  \mathbb{E}_{a_{-i} \sim (\hat{\sigma}_{1 | s^t},\ldots,\hat{\sigma}_{n | s^t})}[u_i(\sigma_i(h_i^t), a_{-i})]   \right)\\
\leq  &\frac{\delta}{1-\delta} \left( \sum_{s \in S, a_i \in A_i}V_{i,\sigma}((h_i^t,s,a_i)) P[s,a_i | h_i^{t}] -  \sum_{s \in S, a_i \in A_i} V_{i,\sigma}((h_i^t,s,a_i)) P[s| h_i^{t}, a^*] P[a_i| \sigma(h_i^t)]  \right)
\intertext{By $(\epsilon,\gamma)$-differential privacy of the private signal, and the fact that $\exp(-\epsilon) \geq 1-\epsilon$, we therefore have that this is at most}
\leq & \frac{\delta}{1-\delta} (\epsilon + \gamma).
\end{align*}
which completes the proof.
\end{proof}
\end{document}